\documentclass{xarticle}

\usepackage{amsfonts,amssymb,amsmath,amsthm}
\usepackage{graphicx}
\usepackage{overpic}
\usepackage[margin=10pt,font=small]{caption}
\usepackage{url}

\def\clap#1{\hbox to 0pt{\hss#1\hss}}
\newcommand{\Z}{\mathbb{Z}}
\newcommand{\R}{\mathbb{R}}
\newtheorem{thm}{Theorem}
\newtheorem{theorem}[thm]{Theorem}
\renewenvironment{proof}{{\noindent\bf Proof. }}{ \hfill ~\qed}
\def\qed{\rule[0pt]{5pt}{5pt}\par\medskip}
\newcommand{\floor}[1]{\lfloor#1\rfloor}\let\mathopfont=\mathrm
\newcommand{\dom}{\mathop{\mathopfont{dom}}}
\newcommand{\ie}{{\it i.e.}}

\usepackage[numbers,sort]{natbib}
\newcommand{\neighbors}{\mathop{\mathopfont{neighbors}}}
\newcommand{\append}{\mathop{\mathopfont{append}}}
\def\ugn{\lambda}
\def\epoch{t^\text{e}}
\def\wu{\omega^\text{u}}
\def\interval#1#2{(#1,#2)}
\def\clap#1{\hbox to 0pt{\hss #1\hss}}
\def\itoj{{ij}}
\def\jtoi{{ji}}

\newcommand{\bittide}{{bittide}}

\begin{document}

\title{Modeling and Control of \bittide\ Synchronization}

\author{Sanjay Lall\footnotesymbol{1}
  \and  C\u{a}lin Ca\c{s}caval\footnotesymbol{2}
  \and  Martin Izzard\footnotesymbol{2}
  \and  Tammo Spalink\footnotesymbol{2}}

\note{Preprint}

\maketitle

\makefootnote{1}{S. Lall is Professor of Electrical
  Engineering at Stanford University, Stanford, CA 94305, USA, and is
  a Visiting Researcher at Google.
  \texttt{lall@stanford.edu}\medskip}

\makefootnote{2}{C\u{a}lin Ca\c{s}caval, Martin Izzard, and Tammo Spalink are
  with Google.}

\begin{abstract}
Distributed system applications rely on a fine-grain common sense of
time.  Existing systems maintain this by keeping each independent
machine as close as possible to wall-clock time through a combination
of software protocols and precision hardware references.  
This approach is expensive,
requiring protocols to deal with asynchrony and its performance
consequences.  Moreover, at data-center scale it is impractical to
distribute a physical clock as is done on a chip or printed circuit
board.  In this paper we introduce a distributed system design that
removes the need for physical clock distribution or mechanisms for
maintaining close alignment to wall-clock time, and instead provides
applications with a perfectly synchronized {\em logical clock}.  We
discuss the {\em abstract frame model} (AFM), a mathematical model
that underpins the system synchronization. The model is based on the
rate of communication between nodes in a topology without requiring a
global clock.  We show that there are families of controllers that
satisfy the properties required for existence and uniqueness of
solutions to the AFM, and give examples.
\end{abstract}

\section{Introduction}

The \bittide\ system is designed to enable synchronous execution
at large scale without the need for a global clock.  Synchronous
communication and processing offers significant benefits for
determinism, performance, utilization, and
robustness.  Synchronization in \bittide\ is decentralized -- every
node in the system adjusts its frequency based on the observed
communication with its neighbors.  This mechanism, first
proposed in~\cite{spalink_2006}, defines a synchronous logical clock
that is resilient to variations in physical clock frequencies.

The design objective is for \bittide\ systems to possess shared
logical time. It is \emph{not} a requirement for this logical time to
perfectly match physical time.  All machines on the network share a
logical discrete clock that ticks in lockstep. This idea is called
\emph{logical synchronization}, to distinguish it from \emph{physical
synchronization}. Viewed from the inside, the behavior of a logically
synchronized system is identical to that of a system with a single
shared physical clock. Viewed from the outside, the logical time is
fully disconnected from physical wall-clock time, meaning that logical
time steps can vary in physical duration, both over time and between nodes.
Applications running on the system use the logical time to
coordinate their actions, which replaces the need to reference
physical time.  Thus, \bittide\ enables perfect \emph{logical}
synchronization, using imperfect \emph{physical} synchronization.

The decentralized nature of the \bittide\ synchronization mechanism
enables building large-scale systems that are synchronized with an
accuracy that is otherwise hard or prohibitively expensive to achieve.
Simply overlaying synchronization information onto asynchronous
communication layers is possible, but in practice has led to large
communication requirements and limited
accuracy~\cite{ntp,corbett_spanner_2013,li_sundial_2020}.  The
\bittide\ system instead achieves synchronization using
\emph{low-level} data flows inherent to serial data links. There is no
communication overhead (in-band signaling) required by the
synchronization mechanism, as the continuous data (meaningful or not)
exchanged at the physical layer is sufficient to provide the necessary
input to our control system.  As we will demonstrate in this paper,
the logical synchronization is accurate, even though the underlying
substrate is only approximately synchronized. This enables building a
much wider class of synchronized systems.

\paragraph{Prior work.}

This particular scheme for synchronization using low-level network
mechanisms originates in~\cite{spalink_2006}. However, other
synchronous network protocols exist, including the heavily-used
SONET~\cite{sonet}. High-level protocols for clock-synchronization
such as NTP are also widely used~\cite{ntp}. Large systems such as
Spanner use proprietary hardware and protocols to keep clocks
as close to each other as
possible~\cite{corbett_spanner_2013,li_sundial_2020}.

The literature discussing synchronization dynamics is large, and we
can only touch upon it briefly here.  The behavior of networked
systems that achieve synchronization via feedback control mechanisms
has been widely studied, and the clock frequency behavior here is
similar to that of several other systems which have been analyzed in
the literature. These include flocking models~\cite{boyds}, Markov
chain averaging models~\cite{hastings,boyd2004}, congestion control
protocols~\cite{kelly}, power networks~\cite{strogatz}, vehicle
platooning~\cite{hedrick}, and flocking~\cite{jadbabaie}. The earliest
work to study such coupled oscillator models of synchronization is
Winfree~\cite{winfree1967}.
Several aspects of the \bittide\ control system are still open
challenges, and we discuss them in Section~\ref{sec:controlobjectives}.

\section{The \bittide\ synchronization mechanism}

We now describe in more detail the structure of
a \bittide\ system. We give here for the sake of clarity the simplest
implementation and omit several possible variations and extensions.

We have a network of computers, represented as an undirected graph,
each node being a computer with a single processor.  Each edge
connects a pair of computers, and corresponds to a pair of
links, one in each direction.
These connections are direct; in this simple case there are no
switches or routers between neighboring nodes.  Bits are sent across
these links grouped into frames. In the simplest case, all frames are
the same (fixed) size.  Implementations may choose the frame size,
which is unrestricted by the \bittide\ system definition.  Because the
frames are of fixed size, determining the boundaries between frames is
straightforward and has low overhead.

Consider two neighboring nodes. At each node there is a queue (called
the \emph{elastic buffer}). Frames are added to the tail of the
elastic buffer as they arrive.  At the head of the elastic buffer,
frames are removed from the buffer and read by the processor. Whenever
a frame is removed from the buffer, a new frame is sent on the
outgoing link back to the sender. Thus each edge between two nodes on
the graph corresponds to \emph{four} objects; a link in each
direction, and an elastic buffer at each node. Because removing a
frame from the head of the queue is always consequent with sending a
new frame on the outgoing link, if we are interested solely in the
network dynamics (and not the actual data in the frames) we can
conceptually view these frames as identical; it would be the same if
each node simply sent back the frames it received, after they
propagated through the elastic buffer. In this sense, the two links
and two buffers form a closed cycle, with frames flowing around
perpetually, as illustrated in Figure~\ref{fig:twonodes}.

\begin{figure}[ht]
  \centerline{\begin{overpic}[width=0.21\textwidth]{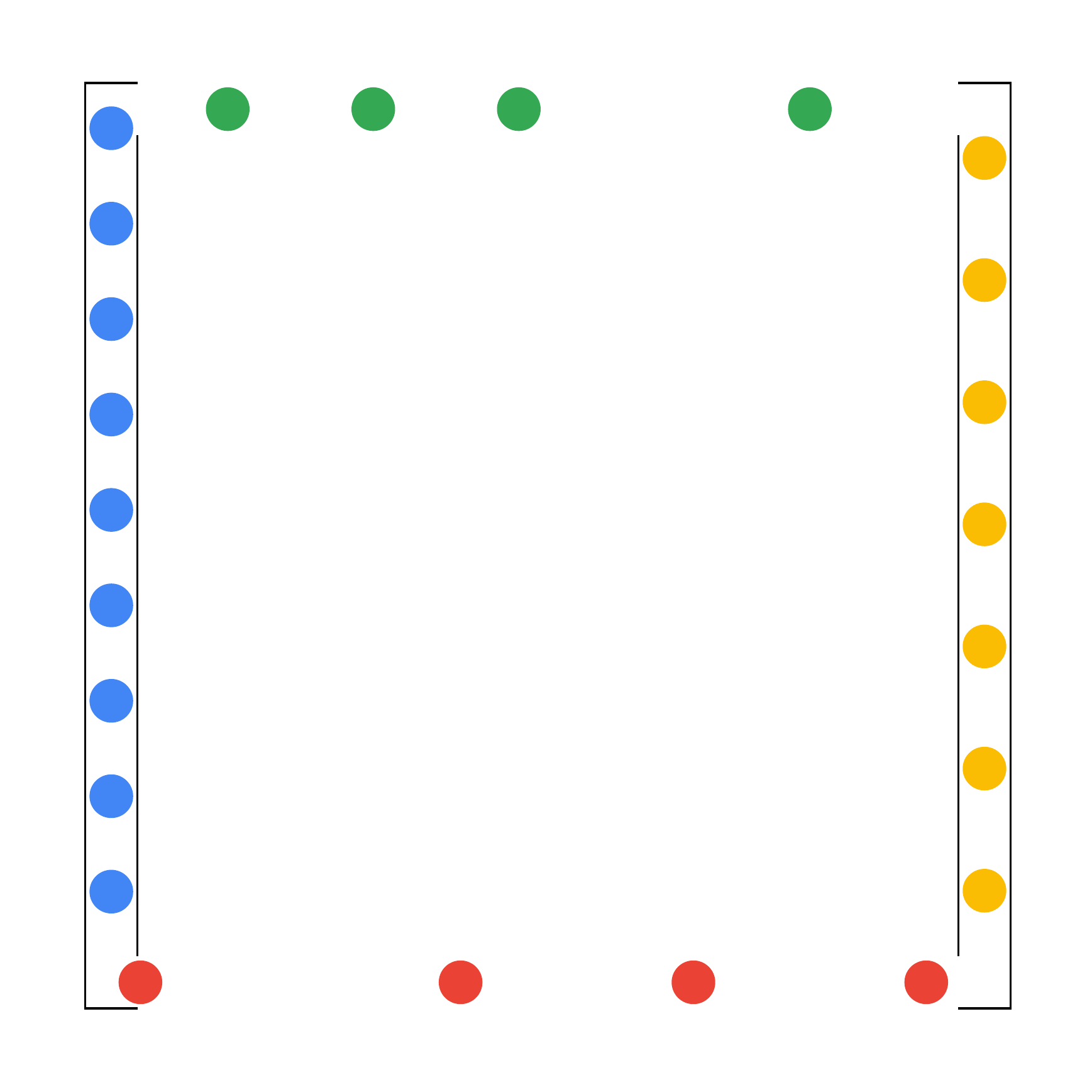}
    \put(50,15){\clap{$\to$ link}}
    \put(12,-4){\clap{$i$}}
    \put(90,-4){\clap{$j$}}
    \put(50,80){\clap{link $\leftarrow$}}
    \put(100,50){\parbox[t]{1in}{elastic\\ buffer}}
    \put(0,50){\llap{\parbox[t]{0.4in}{elastic\\ buffer}}}
    \end{overpic}}
  \caption{Two adjacent nodes.}
  \label{fig:twonodes}
\end{figure}
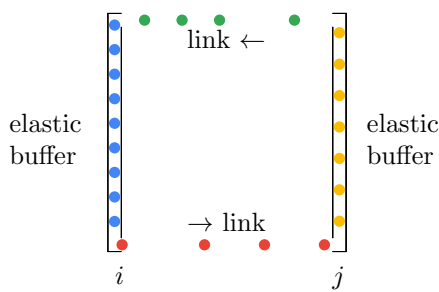

There is one more wrinkle to this behavior. If a node $i$ has, say
$d_i$ neighbors, then it has $d_i$ incoming links (with $d_i$ elastic
buffers), and $d_i$ outgoing links.  The critical feature here is the
timing; frames are sent simultaneously on all outgoing links. It is
worth pointing out that exact simultaneity is not necessary. Instead
we simply need the different transmissions to remain in discrete
lockstep with each other, so that transmissions occur in stages of
$d_i$ frames, one on each link. The next stage does not start until
the previous stage has finished.

We can view each edge as a ring with two buffers and two links, and
frames moving around the ring. Because of the above timing behavior,
each ring turns in lockstep with all other rings. This holds even
though the links may have different physical latency and the buffer
occupancies may vary. The rings behave somewhat analogously to
mechanical gears; even if gears have different sizes and hence
different angular rates, at any point where two gears meet, the teeth
of both gears pass the point at the same rate. 

At each node there is an oscillator, which drives the processor
clock, which on a modern system might run at a few
GHz.  On a \bittide\ system, the oscillator also drives the network
(possibly via frequency multiplication or division), so that
with each clock tick a set of $d_i$ frames is sent, one to each
neighbor. Each network clock tick corresponds to a stage above.

Because the same oscillator is used for the transmission of
frames as for the processor, the lockstep behavior of the network
induces an identical lockstep behavior for all processors. All nodes
coordinate to synchronize these clocks, as will be discussed below. At
a node, with each tick of the clock, a frame is read from each elastic
buffer, a frame is sent on each outgoing link, and a processor
instruction is executed.

The \bittide\ system operates at Layer 1 (physical level) in the OSI
network layer model. This is a significant advantage for
synchronization, as it ensures that there is minimal buffering. This
is in contrast with higher-level protocols such as NTP~\cite{ntp},
where synchronization is at a comparatively coarser level.

The frames that are being transmitted around the network have a dual
purpose; they are used for communication, and they are used for
synchronization.  For the latter, the content of the frames is
irrelevant to the synchronization, and a processor must simply send a
frame whenever it reads one. If there is no data to be communicated,
it must therefore send a frame containing garbage data. (This is not
generally wasteful, because most standard wired network links operate
this way transparently to the user by default.)  Additionally, since
each link is between directly-connected neighbors, it is the
responsibility of higher-level network protocols to perform multi-hop
communication and routing.

The \bittide\ system enables an entire datacenter, and possibly even
larger networks, to operate in logical synchrony. Such behavior is a
significant departure from how current datacenters work, where large
applications run via networks of asynchronous processes.  For many
applications, one of the disadvantages of implementing systems in this
way is \emph{tail-latency}~\cite{dean2013tail}, where delays caused by
a few occasionally-slow processes compound to limit overall
performance.  The \bittide\ system offers a way to eliminate sources of
latency variability via deterministic ahead-of-time scheduling of
processes.

\section{Using feedback control to maintain synchrony}

Consider two neighboring nodes $i$ and $j$. Each node transmits a
frame on an outgoing link whenever it removes a frame from the
corresponding elastic buffer, thus conserving the total number of
frames in the cycle of two links and two buffers.  If the oscillator
of $j$ is faster than that of $i$, then $j$ will pull frames from its
buffer more frequently than they are being supplied, and consequently
its elastic buffer will become emptier. Conversely, $i$'s receive
buffer will become fuller.  The number of frames in each buffer is
called \emph{buffer occupancy}.  Each node has a controller process that
measures buffer occupancy, which provides a feedback signal indirectly
capturing the relative speed of its oscillator compared to the
other. It will then adjust its oscillator speed appropriately.

A node which has $d_i$ neighbors has $d_i$ elastic buffers, which are
all emptied according to the local clock frequency, and filled
according to the frequency of the corresponding neighbor.  The
collection of occupancies of these buffers is therefore a feedback
signal, providing information regarding the relative frequencies of
all of the neighboring nodes.  The question that remains is how best
to adjust the local oscillator frequency in response to these
signals. This is the responsibility of the control algorithm.  
We also need to address the challenge that the
information that each node has about its neighbor is delayed by the
corresponding latency of the physical link; that is, the amount of
time between a frame being sent from the head of the sender's elastic
buffer and it being received into the tail of the recipient's elastic
buffer.  When a node changes its frequency, the effects are felt by
all of its neighbors delayed by the corresponding latencies, and this
in turn leads to the neighbors adjusting their frequencies, and so
on. The effects of frequency adjustment therefore propagate through
the network. As a result of these dynamics, both the latencies and the
graph topology can have a significant effect on the overall behavior
of our decentralized synchronization mechanism.

\section{Modeling the dynamics of frames}

Let the set of all nodes in the graph be
$\mathcal{V} = \{1,2,\dots,n\}$. Each node has a clock, whose value at
any time $t$ is a real number $\theta_i(t)$, called the \emph{phase}
of the clock.  It is driven by an oscillator whose frequency is
$\omega_i(t)$, and so satisfies the dynamics
\[
\dot\theta_i(t) = \omega_i(t) \quad \text{for all } i \in \mathcal V
\]
The value of $\theta_i(t)$ is called the \emph{local time} at node
$i$, or the time in \emph{local ticks}.  The frequency $\omega_i(t)$
may vary over time, both as a result of physical variations such as
temperature, and as a result of adjustments by the controller at node
$i$.  This clock drives the frame reading and transmission processes
in the following way.  Every time~$t$ at which the clock $\theta_i(t)$
is an integer, node $i$ removes a frame from the head of the
corresponding elastic buffer and sends a frame on the link from $i$ to
$j$. Therefore, if~$\omega_i$ does not change over time, node $i$
sends $\omega_i$ frames per second.

This simple model is enough to determine the location of all of the
frames within the system, as follows. Suppose $s < t$, then the number
of frames that have been sent on the time interval $(s,t]$ by node $i$
is
\[
\sigma_i\interval{s}{t} = \floor{\theta_i(t)} - \floor{\theta_i(s)}
\]
Between any two nodes $i$ and $j$, there are two links, one in each
direction.  The link from $i$ to $j$ has latency $l_\itoj$.  Note that
the latency includes the time to serialize a frame, transmit it across
the physical link, and deserialize the frame into the elastic buffer.
It does not include the time for the frame to propagate from the tail
to the head of the elastic buffer, and so is not a measure of
communication delay between nodes.

On this link, the number of frames at time~$t$ is therefore
\begin{align*}
  \gamma_\itoj(t) &= \sigma_i\interval{t-l_\itoj}{t} \\
  &= \floor{\theta_i(t)} - \floor{\theta_i(t-l_\itoj)}
\end{align*}
Note that this implies a specific interpretation of boundary points,
so that frames which at time $t$ are exactly at the start of the link
are considered as on the link, and frames that are at the end of the
link are considered as no longer on the link.

Define the number of frames received into the elastic buffer at $j$
from $i$ over the interval $(s,t]$ to be $\rho_\itoj\interval{s}{t}$.
  Since the links do not drop frames, the number received is simply
  the number sent, delayed by the latency. We have
\[
\rho_\itoj\interval{s}{t} = \sigma_i\interval{s-l_\itoj}{t-l_\itoj}
\]
To determine the occupancy of the elastic buffer, we need to specify
additional information, because while knowledge of $\theta$ informs us
of how many frames arrive and leave the buffer within a particular
time interval, the total number of frames in the buffer also depends
on how many it contained beforehand. So we define $\beta_\itoj(t)$ to
be the occupancy of the elastic buffer at node $j$ associated with the
link from node $i$ at time $t$, and $\beta^0_\itoj$ to be the
occupancy at time $t=0$. The following gives a mathematical definition
of $\beta_\itoj$.  We construct $\beta_\itoj$ as the unique function
for which
\[
\beta_\itoj(0) = \beta_\itoj^0
\]
and for which the following difference relationship holds for all $t,s
\geq 0$.  The difference between the occupancy at time $t$ and the
occupancy at time $s$ is simply the number of frames received minus
the number sent over that interval.
That is,
\begin{align}
  \beta_\itoj(t) - \beta_\itoj(s)  \hskip -40pt &  \nonumber \\
  &= \underbrace{\rho_\itoj\interval{s}{t}}_\text{received} -
  \underbrace{\sigma_j\interval{s}{t}\vphantom{\rho_i}}_\text{sent\vphantom{d}} \nonumber \\
  &=
  \label{eqn:betadiff}
  \sigma_i\interval{s-l_\itoj}{t-l_\itoj} -\sigma_j\interval{s}{t} \\
  &=
  \floor{\theta_i(t-l_\itoj)} - \floor{\theta_i(s-l_\itoj)}
  - \floor{\theta_j(t)} + \floor{\theta_j(s)} \nonumber
\end{align}
Now define
\begin{equation}
  \label{eqn:ugn}
  \ugn_\itoj(t) =  \beta_\itoj(t)  -\floor{\theta_i(t-l_\itoj)}
  +  \floor{\theta_j(t)}
\end{equation}
and notice that equation~\eqref{eqn:betadiff} implies that
\[
\ugn_\itoj(t) = \ugn_\itoj(s) \qquad \text{for all }s,t
\]
Hence $\ugn_\itoj$ is constant.  Evaluating equation~\eqref{eqn:ugn} at
$t=0$ shows that it may be determined from $\beta_\itoj^0$ and the
initial conditions for $\theta$, which specify the value of
$\theta(t)$ for all $t \leq 0$. Note that the initial conditions are
not simply given by $\theta(0)$ since the system contains delays, and
is therefore infinite-dimensional. Then we have
\[
\beta_\itoj(t)
 = \floor{\theta_i(t-l_\itoj)} - \floor{\theta_j(t)} + \ugn_\itoj
 \]
Note in particular that $\ugn_\itoj$ and $\ugn_\jtoi$ may differ.

Using the above definition of $\gamma_\itoj$, it is convenient to
write $\ugn_\itoj$ as
\[
\ugn_\itoj = \beta_\itoj(t) + \gamma_\itoj(t) + \floor{\theta_j(t)} - \floor{\theta_i(t)}
\]
This number is the buffer occupancy plus the link occupancy, plus the
clock offset between the nodes. We can interpret its invariance as
follows.  It is constant because the first two terms sum to the number
of frames on the path from the head of the buffer at $i$ to the head
of the buffer at $j$.  The only way this can change is via one or
other clock increasing by one, and those two actions correspond to a
frame being added or removed from this path. Furthermore, we have
\[
\ugn_\jtoi + \ugn_\itoj = \beta_\jtoi(t) + \gamma_\jtoi(t)
+ \beta_\itoj(t) + \gamma_\itoj(t)
\]
which means that the total number of frames on the two links plus two
buffers is conserved.

In the above description we have for simplicity only discussed the
case where all links transmit at the same rate. However, it is
straightforward to extend this model to include links which send
at different rates. To do this, one adds \emph{gearboxes} $g_\itoj$
so that node $i$ sends $g_\itoj$ frames onto link $i\to j$ for every tick
of $\theta_i$. Then the buffer occupancies become
\[
\beta_\itoj(t) =  \floor{g_\itoj\theta_i(t-l_\itoj)}
- \floor{g_\itoj \theta_j(t)}  + \ugn_\itoj
\]
Since this additional complexity does not affect the control mechanism
responsible for synchronization we do not discuss it further here.

In a hardware implementation, each node has memory dedicated to the
elastic buffer for each network interface.  Two pointers are stored
which keep track of each end of the buffer, so that adding or removing
a frame does not require data to be moved in memory. However, the
buffer has a fixed size, and so can overflow. Both overflow and
underflow at any node are fatal errors for the \bittide\ system.

The requirement that the elastic buffers neither overflow nor
underflow means that the difference in clock frequencies at the two
ends of a link cannot stay too large for too long; if it does, either
the buffer at the low-frequency end will overflow or the buffer at the
high-frequency end will underflow, or both.

Define the elastic buffer length to be $\beta^\text{max}$. Then we
must ensure that the frequencies $\omega$ are such that the
occupancies $\beta_\itoj$ satisfy
\[
0 \leq \beta_\itoj(t) \leq \beta^\text{max}
\text{ for all } t\geq 0 \text{ and } i,j \in \mathcal V
\]
This is the fundamental performance requirement that the control
system must enforce. Additionally, it is preferable that
$\beta_\itoj(t)$ be small, since smaller buffers mean smaller
communication latency. Here, by communication latency we mean the
amount of real time (wall-clock time) that it takes for a frame to
leave the head of the source elastic buffer and arrive at the head of
the destination elastic buffer.

Achieving this requirement means we must have clocks that are
operating (on average) at the same frequency at all nodes. In
practice, left alone, no two clocks will remain perfectly
synchronized, and over time their counters will diverge.  Some of the
most stable clocks are atomic clocks, which offer a relative error of
about 1 part in $10^{11}$.  This much error means that the buffer will
accumulate about 1 bit every 100 seconds on a gigabit link.  To avoid
buffer overflow and underflow, we therefore need to use feedback
control to stabilize the buffer occupancies.

\section{Connecting a controller}

The basic idea of the control system is that it can measure the buffer
occupancies $\beta$ and set the frequencies $\omega$. However, there
are several complicating factors. The most fundamental is that the
controller cannot actually set the exact frequency $\omega$. The
oscillator at node $i$ has a frequency at which it will operate 
when it is not corrected, which is called the \emph{uncorrected frequency},
denoted by $\wu_i$.  When the oscillator is controlled, the
frequency $\omega_i$ is
\[
\omega_i(t) = c_i(t) + \wu_i(t)
\]
where $c_i$ is the \emph{correction} set by the controller.  The
uncorrected frequency is typically not known exactly, and is subject
to both manufacturing tolerances and the effects of aging,
temperature, and other physical effects.
It changes with time, and is not
measurable while the system is running. Models for the change in $\wu$
over time include phenomena such as drift and
jitter~\cite{allan1987time}.

At each node, by design, there is no way to measure wall-clock time $t$; the best
the processor can do is observe the local clock $\theta_i$. It is this
notion of time that determines when the occupancy of the elastic
buffers is sampled and when the frequency is updated. Therefore, the
sample-rate at each node varies, depending on the state of the
system. This is one of the sources of nonlinear behavior in the
system.

The model of the system is written using wall-clock time $t$ as
independent variable.  However, we do not assume that the controller
can observe $t$. Node $i$ measures the buffer occupancies
at times $t=t_i^0, t_i^1, t_i^2, \dots$. These sample times are
defined by
\[
\theta_i(t_i^k) = \theta_i^0 + kp
\]
Here $p\in\Z_+$ is the sample period, in local ticks. Notice that,
since the initial conditions specify that $\theta_i(0)= \theta_i^0$,
the first sample time is $t_i^0 = 0$. After sampling at time
$t=t^k_i$, node $i$ sets the frequency correction $c_i$ at a
time~$d$ local ticks later.  Specifically, the correction is set at times
$t=s_i^0, s_i^1, s_i^2, \dots$, defined by
\[
\theta_i(s_i^k) = \theta_i^0 + kp + d
\]
It is important that the initial phase $\theta_i^0$ is not an
integer. Otherwise, the buffer occupancy is measured at exactly the
times when $\theta_i$ is integral, and those are precisely the times
at which a frame is removed from the elastic buffer. While this is
mathematically well defined, in practice we cannot measure buffer
occupancy exactly at this time. The interpretation of the fractional
part $\theta_i^0 - \floor{\theta_i^0}$ is that it specifies when the
samples are made, relative to the removal of frames from the buffer.

At each time $t_i^k$, the controller at node $i$ measures $y_k^i$, the set
of buffer occupancies at that node,
that is
\[
y^k_i = \{ (j, \beta_\jtoi(t_i^k)) \mid j \in \neighbors(i)\}
\]
Each buffer occupancy is labeled with the neighboring node $j$ that
supplies it. The controller is a function $\chi_i^k$ which
maps the history of these measurements to the correction $c_i^k\in\R$,
according to
\begin{equation}
  \label{eqn:ioctrl}
  c_i^k = \chi_i^k(y_i^0, \dots, y_i^k)
\end{equation}
This correction is applied on the interval $[s_i^k,
s_i^{k+1})$, and
as a result the frequency is \emph{piecewise constant}, with
\begin{equation}
  \dot\theta_i(t) = c_i^k + \wu_i \quad \text{for } t \in[s_i^k,s_i^{k+1})
\end{equation}

\subsection{The abstract frame model}

The model defined by the above equations is called the \emph{abstract
frame model} for \bittide.  It defines the connection between the
controller and the clock, using an ideal abstraction for the frames,
by which one can determine the location of all of the frames using
only the history of~$\theta$. Since this is a transport model, one
might also model it using an advection partial differential equation,
but here we use the delay-differential equation form.  We summarize
the model here. For all $t \geq0$, $i \in \mathcal V$, and~$k\in\Z_+$,
\begin{equation}
  \label{eqn:afm}
  \begin{aligned}
  \dot\theta_i(t) &= c_i^k + \wu_i \quad \text{for } t \in[s_i^k,s_i^{k+1}) \\
  \beta_\jtoi(t)   &= \floor{\theta_j(t-l_\jtoi)} - \floor{\theta_i(t)} + \ugn_\jtoi\\
  \theta_i(t_i^k) &= \theta_i^0 + kp  \\
  \theta_i(s_i^k) &= \theta_i^0 + kp + d \\
  y^k_i &= \{ (j, \beta_\jtoi(t_i^k)) \mid j \in \neighbors(i)\} \\
  c_i^k &= \chi_i^k(y_i^0, \dots, y_i^k)
  \end{aligned}
\end{equation}
The initial conditions of the model are
\begin{equation}
  \label{eqn:ics}
  \theta_i(t) =
  \begin{cases}
    \theta_i^0 + \omega_i^{(-2)}t & \text{for } t \in [\epoch, 0] \\
    \theta_i^0 + \omega_i^{(-1)}t & \text{for } t \in [0, d/\omega_i^{(-1)}]
  \end{cases}
\end{equation}
The first of these equations ensures that the dynamics are
well-defined, by specifying the initial value of $\theta$ on a
sufficiently large interval of time for the delay dynamics. The second
equation defines $\theta$ on the period between time zero and the time
at which the first controller action takes effect.

A controller is called \emph{admissible} if
\[
\chi_i^k(y_i^0, \dots, y_i^k) + \wu_i >  \omega^{\text{min}}
\]
for all $i,k,y_i^0,\dots,y_i^k$. This ensures that
\begin{equation}
  \label{eqn:freqbound}
  \omega_i(t)  > \omega^{\text{min}}
\end{equation}
The parameters of the model are as follows. The minimum frequency is
$\omega^{\text{min}}>0$. The \emph{epoch} is $\epoch<0$, and it
must satisfy
\[
\epoch \leq -(l_\itoj + d/\omega^\text{min}) \text{ for all } i,j \in \mathcal V
\]
The initial buffer occupancies are $\beta_\itoj^0 \in\Z_+$. The
sampling period is $p \in\Z_+$, the number of clock cycles between
controller updates. The delay $d\in\Z_+$ models the time (in local
ticks) taken to compute the controller update and for the oscillator
to respond to a frequency change.  We assume $d<p$.  The initial
frequencies are
\[
\omega_i^{(-1)} > \omega^\text{min} \text{ and }
\omega_i^{(-2)}> \omega^\text{min}
\]
The initial clock phases are $\theta^0_i \in\R^+ \backslash \Z$. The uncorrected
oscillator frequencies are $\wu_i\in\R_+$.
The constants
$\ugn_\itoj$ are computed according to
\[
\ugn_\itoj =  \beta_\itoj^0  -\floor{\theta_i(-l_\itoj)}  +  \floor{\theta_j^0}
\]
With a state-space decentralized controller, we have
\begin{align*}
  \xi_i^{k+1} &= f_i(\xi_i^k, y_i^k) \quad \\
  c_i^k &= g_i(\xi_i^k, y_i^k)
\end{align*}
where $\xi_i^k$ is state of the controller at node $i$ and step~$k$. This results
in an input-output controller map of the form~\eqref{eqn:ioctrl}.

\section{Existence and uniqueness of solutions}

It is important to establish the existence and uniqueness of solutions
for the abstract frame model. The model's dynamics are a type of
hybrid system, with nonlinearities, state-dependent multi-rate
sampling and delays, making the analysis challenging. In addition, the
buffer occupancies are discrete, and this allows for the possibility
that the dynamics will exhibit Zeno behavior. We will establish that
admissibility~\eqref{eqn:freqbound} is a critical condition to avoid
Zeno behavior when frames enter or leave the elastic buffer. As such,
we can ensure the existence and uniqueness of solutions for any
controller that satisfies the assumptions above.

We summarize the controller and sampling behavior of the model by writing it as
\begin{align*}
  \dot\theta_i(t) &= c_i^k + \wu_i && \text{for } t \in[s_i^k,s_i^{k+1}) \\
    \theta_i(s_i^k) &= \theta_i^0 + kp + d \\
  c_i^k &= G_i(\theta, s_i^k)
\end{align*}
Here the function $G_i$ contains the construction of the sample times
$t_i^k$, the measurement, and the controller. Notice that the first
argument of $G_i$ is the entire history of $\theta$, not just its
value at a particular time. To state this precisely, define the
(non-minimal) state space for the system as follows. Consider
functions $f:[\epoch, b] \to \R$ with $b>0$ or $f:[\epoch, \infty) \to
  \R$, which are piecewise linear, continuous, and satisfy
\[
f'(t) > \omega^\text{min} \text{ for almost all } t\in\dom f
\]
Let $\mathcal P$ denote the set of all such functions, and $\mathcal X
= \mathcal P^n$.  The function $G_i$ has domain $\mathcal D \subset
\mathcal X \times \R$, and $G_i:\mathcal D \to \R$.  Specifically, for
$\theta \in \mathcal X$, $i\in\mathcal V$, and $s\in\R$, we have
$\theta,s \in \mathcal D$ if $s\in\dom \theta$ and $\theta_i(s) =
\theta_i(0) + kp + d$ for some $k\in\Z_+$.  Under these conditions,
for $l=0,\dots,k$, let $t^l$ be such that
\[
\theta_i(t^l) = \theta_i(0) +lp
\]
which must exist, since $\theta_i$ is strictly increasing. Evaluate
\begin{align*}
  y^l &= \{ (j, \beta_\jtoi(t^l)) \mid j \in \neighbors(i)\} \\
  c &= \chi_i^k(y^0, \dots, y^k)
\end{align*}
and define $G_i(\theta,s) = c$. Notice that evaluating $\beta_\jtoi$
requires evaluating $\theta$ at times $-l_\jtoi$, which lie within the
domain of $\theta$ by the requirements on $\epoch$.

We can now prove the following result.
\begin{theorem}
  There exists a unique $\theta\in\mathcal{X}$ satisfying the abstract
  frame model~\eqref{eqn:afm} and initial conditions~\eqref{eqn:ics}.
\end{theorem}
\begin{proof}
  First let $\theta\in\mathcal X$ be defined by initial
  conditions~\eqref{eqn:ics}. We can now define a map $F:\mathcal X
  \to \mathcal X$ as follows. Given $\theta \in \mathcal X$, define
  $i$ according to
  \begin{equation}
    \label{eqn:choice}
    i = \min \arg \min_i \max \dom \theta_i
  \end{equation}
  which simply finds the component of $\theta$ which has the smallest
  domain, breaking ties by choosing the smallest index. Let $s=
  \max\dom \theta_i$. We have $\theta,s$ lies in $\dom G_i$, and so
  let $c=G_i(\theta,s)$. We now extend the piecewise linear function
  $\theta_i$ by adding a point so that
  \[
  \theta_i(s + p/(c+\wu_i)) = \theta_i(s)+p
  \]
  and extending $\theta_i$ to this point via linear interpolation.
  Let this newly extended function be $\theta^+$, whose derivative is
  by construction bounded below by $\omega^\text{min}$.  We define~$F$
  by $\theta^+ = F(\theta)$.  The proof now proceeds by induction,
  constructing a sequence of functions by repeatedly applying~$F$ to
  $\theta$. At each step, the domain of one component of $\theta$ is
  extended by at least $p/\omega^\text{min}$.  This is always the
  component with the smallest domain, and so in the limit the domain
  of $\theta$ extends to $[\epoch,\infty)$ as desired. Uniqueness
  follows by observing that the order in which updates are
  performed, as determined by the breaking of ties
  in~\eqref{eqn:choice}, does not affect the resulting solution.
\end{proof}

\section{Simulation}

The above proof of existence also leads immediately to the following algorithm
for simulating the system.

\begin{tabbing}
  \hskip 20pt \=  $s \leftarrow 0$ \\[1mm]
  \> $\xi_i \leftarrow \xi_i^0 \text{ for all } i \in \mathcal V$ \\[1mm]
  \> $\theta_i \leftarrow \text{initial conditions of~\eqref{eqn:ics}} \text{ for all } i \in \mathcal V$ \\[1mm]
  \> while $s < t_\text{max}$ \\[1mm]
  \> \hskip 20pt \=
  $i \leftarrow \min \arg \min_i \max \dom \theta_i$\\[1mm]
  \> \> $s \leftarrow \max \dom \theta_i$ \\[1mm]
  \> \> $t \leftarrow \theta_i^{-1}(\theta_i(s) - d)$ \\[1mm]
  \> \> $y \leftarrow \{ (j, \beta_\jtoi(t) \mid j \in \neighbors(i) \}$ \\[1mm]
  \> \> $\xi_i \leftarrow f_i (\xi_i, y)$ \\[1mm]
  \> \> $c \leftarrow g_i (\xi_i, y)$ \\[1mm]
  \> \> $\append(\theta_i, (s+p/(c+\wu_i),  \theta_i(s) + p) )$
\end{tabbing}

Here $\theta_i$ is a piecewise linear increasing function, stored as a
list of knots, \ie, pairs of real numbers. It can be evaluated via
linear interpolation, as can the inverse function $\theta_i^{-1}$. The
$\append$ function simply adds a knot to the end of the list.  Also
$\max \dom \theta_i$ is given by the independent variable of the last
knot.

Code to simulate this system is available at \linebreak \url
{https://bittide.googlesource.com/callisto}.

\subsection{Limitations of the model}

The abstract frame model allows one to determine the location of all
of the individual frames on the network. For control systems, one
often uses models at a coarser level of precision.  The AFM does,
however, omit certain phenomena. In particular, the model assumes that
a frame is inserted into the elastic buffer as soon as it has
traversed the link.

The model also does not include limitations which may be imposed by
the particular physical oscillator used. The oscillator frequency is
set by writing the desired frequency offset into a register, the
number of bits of which determines the number of quantization levels
available for control. Depending on the specific system parameters,
this quantization may have a significant effect on the control system.

\section{Control objectives and  requirements} \label{sec:controlobjectives}

The objective of \bittide\ is the maintenance of logical
synchronization. No matter what frequencies the nodes run at, the
logical synchronization happens as a consequence of the lockstep
behavior of the system --- until, that is, the buffer overflows, or
underflows.  Thus, the primary objective of the control system is to
manage the buffer occupancy --- keep it within limits and make it as
small as possible.  Buffer occupancy translates directly into
communication latency, therefore smaller is better, as long it it does
not underflow. Keeping the buffer within limits also requires
frequencies to not deviate too much, or for too long from each other.
Minimizing the frequency deviation is not a goal.  However,
the physical oscillator may drift, and the controller should tolerate
this drift.

A secondary objective is to keep the frequency as large as possible;
if there were no frequency limits, a trivial control strategy for
managing buffer occupancy would simply be to reduce all nodes to close
to zero frequency. Since the processor cores are clocked exactly in
lockstep with frame transmission, higher frequencies are better
because they result in faster computation.  The controller should
behave well when encountering constraints imposed by frequency limits.

All of the above objectives must be achieved with a
\emph{decentralized controller.}  This has profound consequences for
the application layer, enabling strong security guarantees. For the
control algorithm, it means that the individual controllers cannot
communicate directly with each other to share measurement information
and coordinate their actions.

We impose specific design choices on the controllers for a
\bittide\ system, motivated by the desire for strong isolation.  They
cannot use Paxos-style consensus algorithms, or elect a leader. They
cannot mark or inspect frames.  They cannot observe the real time
$t$. They cannot broadcast information on the network.  Each node must
update its frequency using only observations of its elastic buffers.

There is however an additional freedom available to
the \bittide\ controller. After booting, once the buffer occupancies and frequencies
have reached equilibrium, each node may adjust its buffer, discarding
frames or adding frames, and thereby reset its buffer occupancy to a
desired value.  This can happen only at bootup, because at this stage
the frames do not yet contain application data.  Therefore the
objective that buffer occupancies should be regulated applies only
after the initial transients of the boot phase.

The \bittide\ control system presented here measures buffer occupancy
and uses this directly to choose the frequency correction. More
sophisticated schemes which estimate frequencies at neighbors are possible
and may offer performance benefits.

There are further practical considerations for the \bittide\ system
controller.  These include allowing nodes to leave or join the network
gracefully, and detection of node and link failures.  It must also
work over a wide range of network topologies, and it is preferable
that minimal configuration should be necessary to inform the
controller of the topology and link latencies.  The controller should
handle a broad range of frequencies.  There are many ways to formalize
these objectives, and the design of a controller that achieves all of
these objectives remains a subject for research.

\section{Example}

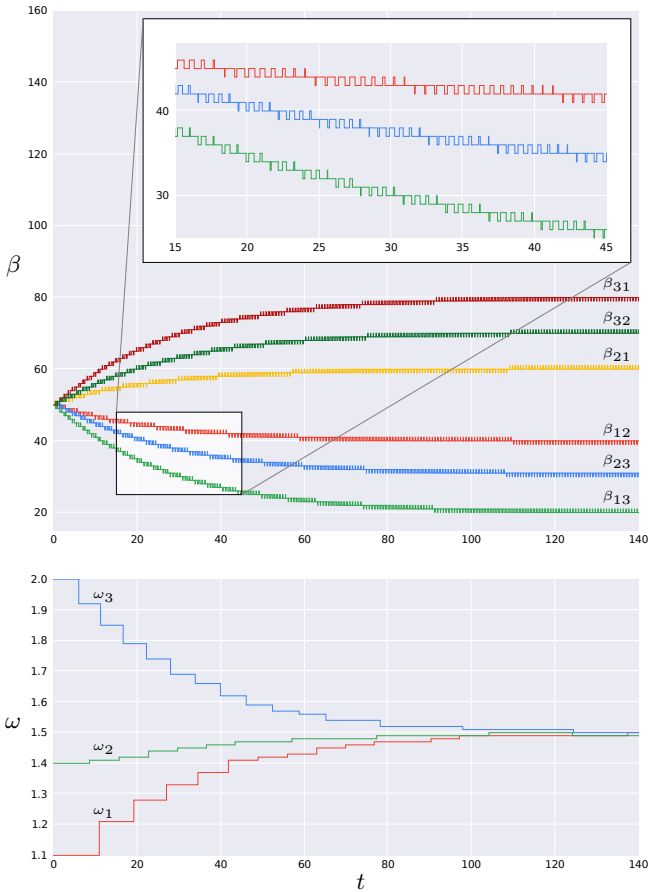
\begin{figure}[t]
  \centerline{\begin{overpic}[width=1\linewidth]{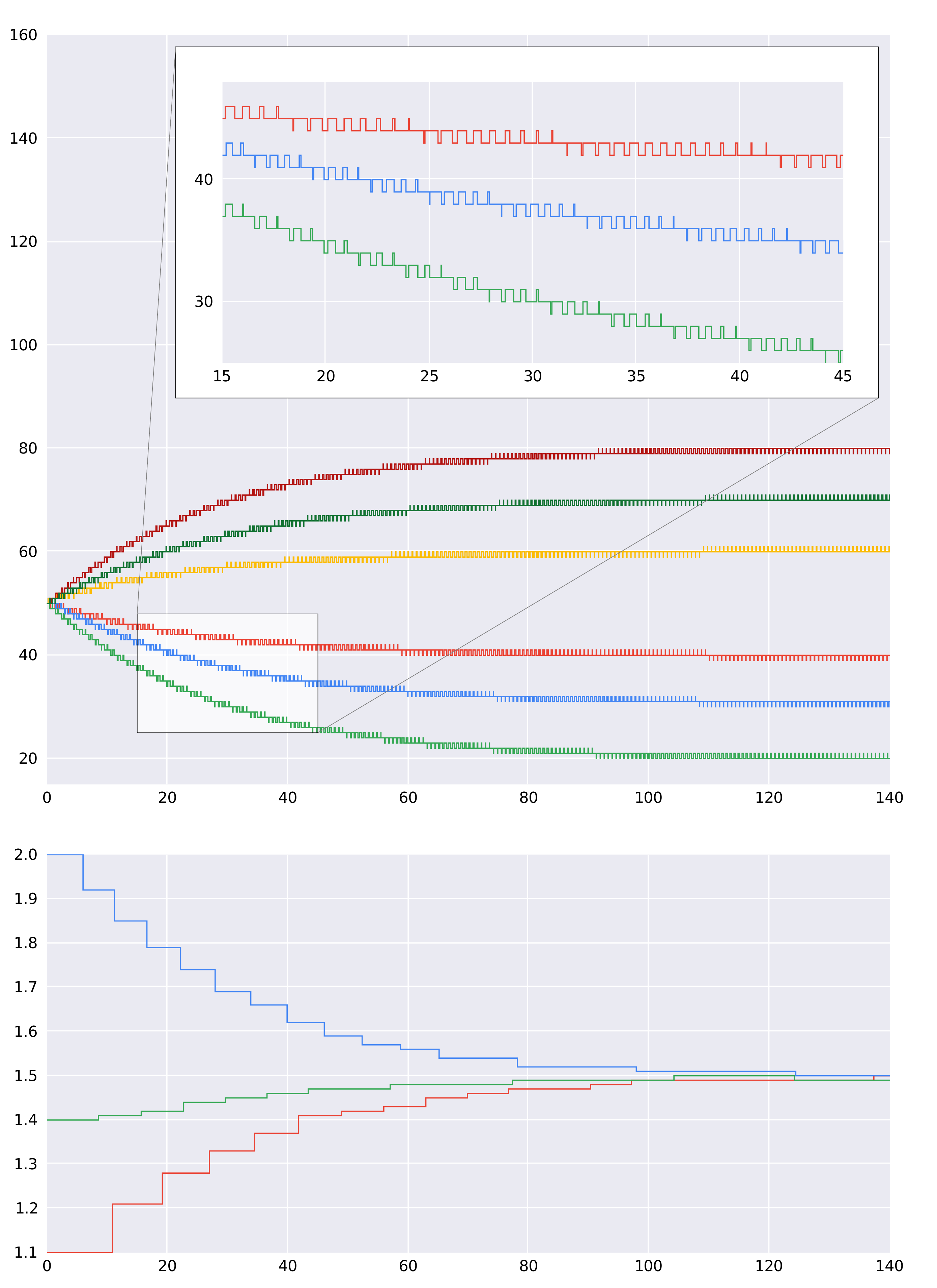}
      \put(0,17){\llap{\small$\omega$}}
      \put(0,68){\llap{\small$\beta$}}
      \put(38,-1){\clap{\small$t$}}
      \put(65,66.3){\tiny$\beta_{31}$}
      \put(65,62.5){\tiny$\beta_{32}$}
      \put(65,58.5){\tiny$\beta_{21}$}
      \put(65,50){\tiny$\beta_{12}$}
      \put(65,46.5){\tiny$\beta_{23}$}
      \put(65,42.5){\tiny$\beta_{13}$}
      \put(8,31.5){\tiny$\omega_3$}
      \put(8,14.5){\tiny$\omega_2$}
      \put(8,7.3){\tiny$\omega_1$}
  \end{overpic}}
  \vspace*{3mm}
  \caption{Occupancy and frequency for a system with three nodes.}
  \label{fig:threenodes}
\end{figure}

Figure~\ref{fig:threenodes} shows an example simulation for a graph
with three nodes and three edges in a triangular topology. The parameters
for this system are
\[
p = 10, \quad d = 2, \quad l_\jtoi=1, \quad \theta^0_{i} = 0.1,
\quad \beta_\jtoi^0 = 50
\]
for all $i,j$, and the uncorrected frequencies are
\[
\wu = (1.1, 1.4, 2.0)
\]
These parameters are chosen for illustrative purposes only (so that
the variations in occupancy and frequency are visible in the figure).
For example, on a modern system typically the uncorrected frequencies would
be much closer together. The controller used here is a simple
proportional controller, given by
\[
c_i = k_P \sum_{j \in \neighbors(i)} \beta_\jtoi
\]
and the gain $k_P=0.01$. A stability analysis of this system can be
found in~\cite{ecc}.

Some features of this mechanism are apparent. The frequency at a node
is proportional to the sum of the buffer lengths, and so short buffer
lengths at a node will cause that node to have a low rate of
transmission, and so its buffer lengths will increase. Consequently we
expect that frequencies and buffer lengths will tend to equilibrate,
as seen in these plots, and that at equilibrium all frequencies will
be close to each other. We can also observe that the buffer
occupancies at each end of a link almost sum to a constant. This
mirroring of the buffer occupancies is not perfect, due to the effects
of latency.

\section{Conclusions}

In this paper we have presented a model for the \bittide\ synchronization
mechanism. We have discussed the unique features and
requirements of the control design problem. As these systems are
deployed, we anticipate further research will be developed for these
systems.

\section{Acknowledgments}

We thank Sam Grayson, Sahil Hasan, Sarah Aguasvivas Manzano, Jean-Jacques
Slotine, and Tong Shen for many fruitful 
discussions. We particularly thank Sarah Aguasvivas Manzano for
carefully reading the manuscript.

\end{document}